\documentclass[sigconf]{acmart}
\setcopyright{none}

\usepackage{enumitem}
\usepackage{nicefrac}
\usepackage{url}

\usepackage[utf8]{inputenc}
\usepackage{graphicx}
\usepackage{latexsym}
\usepackage{amsmath}
\usepackage{url}
\usepackage{float}
\usepackage{tikz}
\usepackage{fancyvrb}
\usepackage{listings}
\usepackage{xcolor}
\usepackage{color}
\usepackage{soul}
\usepackage{multirow}
\usepackage{multicol}
\usepackage{subfig}
\usepackage{booktabs}
\usepackage{caption}

\copyrightyear{}
\acmYear{}
\setcopyright{none}
\acmConference{}
\acmBooktitle{}
\acmPrice{}
\acmDOI{}
\acmISBN{}

\usepackage{mathtools}
\usepackage{algorithm}
\usepackage{algorithmicx}
\usepackage{algpseudocode}
\usepackage{subfig}

\newtheorem{observation}{Observation}
\algdef{SE}[DOWHILE]{Do}{doWhile}{\algorithmicdo}[1]{\algorithmicwhile\ #1}%
\DeclarePairedDelimiter\ceil{\lceil}{\rceil}

\renewcommand{\vec}[1]{\boldsymbol{\mathrm{#1}}}

\DeclareMathOperator*{\maximize}{maximize}
\DeclareMathOperator{\argmax}{argmax}
\DeclareMathOperator{\argmin}{argmin}

\providecommand{\vx}{\ensuremath{\vec{x}}}

\PassOptionsToPackage{hyphens}{url}\usepackage{hyperref} 
\definecolor{mylinkcolor}{RGB}{0,0,140}
\hypersetup{colorlinks,allcolors=mylinkcolor,citecolor=mylinkcolor}
\usepackage[capitalize]{cleveref}

\begin{document}
\title{The Generalized Mean Densest Subgraph Problem}

\begin{abstract}
	Finding dense subgraphs of a large graph is a standard problem in graph mining that has been studied extensively both for its theoretical richness and its many practical applications. In this paper we introduce a new family of dense subgraph objectives, parameterized by a single parameter $p$, based on computing generalized means of degree sequences of a subgraph. Our objective captures both the standard densest subgraph problem and the maximum $k$-core as special cases, and provides a way to interpolate between and extrapolate beyond these two objectives when searching for other notions of dense subgraphs. 
	In terms of algorithmic contributions, we first show that our objective can be minimized in polynomial time for all $p \geq 1$ using repeated submodular minimization. A major contribution of our work is analyzing the performance of different types of peeling algorithms for dense subgraphs both in theory and practice. We prove that the standard peeling algorithm can perform arbitrarily poorly on our generalized objective, but we then design a more sophisticated peeling method which for $p \geq 1$ has an approximation guarantee that is always at least $1/2$ and converges to 1 as $p \rightarrow \infty$. In practice, we show that this algorithm obtains extremely good approximations to the optimal solution, scales to large graphs, and highlights a range of different meaningful notions of density on graphs coming from numerous domains. Furthermore, it is typically able to approximate the densest subgraph problem better than the standard peeling algorithm, by better accounting for how the removal of one node affects other nodes in its neighborhood.
\end{abstract}

\author{Nate Veldt}
\affiliation{
	\institution{Cornell University}
	\department{Center for Applied Mathematics}
}
\email{nveldt@cornell.edu}

\author{Austin R.~Benson}
\affiliation{%
	\institution{Cornell University}
	\department{Department of Computer Science}
}
\email{arb@cs.cornell.edu}

\author{Jon Kleinberg}
\affiliation{%
	\institution{Cornell University}
	\department{Department of Computer Science}
}
\email{kleinberg@cornell.edu}

\maketitle


\section{Introduction}
Detecting densely connected sets of nodes in a graph is a basic graph mining primitive~\cite{gionis2015dense,lee2010survey}. The problem is related to graph clustering, but differs in that it only consider the internal edge structure of a subgraph, and not the number of external edges connecting it to the rest of the graph. Algorithms and hardness results for dense subgraph discovery have been studied extensively in theory~\cite{charikar2000greedy,lee2010survey,andersen2009finding} 
and applied in a wide array of applications, including 
discovering DNA motifs~\cite{fratkin2006motifcut}, 
finding functional modules in gene co-expression networks~\cite{hu2005mining}  and communities in social networks~\cite{sozio2010community},
identifying trending topics in social media~\cite{angel2012dense},
classifying brain networks~\cite{lanciano2020explainable}, and
various other data mining tasks~\cite{sariyuce2018,tsourakakis2015k,shin2018patterns,zhang2017hidden}.

The typical approach for detecting dense subgraphs is to set up and solve (or at least approximate) a combinatorial optimization problem that encodes some measure of density for each node set in a graph. One simple measure of density is the number of edges in a subgraph divided by the number of pairs of nodes. This density measure plays an important role in certain clustering objectives~\cite{veldt2018correlation}, but by itself it is not meaningful to optimize, since a single edge maximizes this ratio. If, however, one seeks a maximum sized subgraph such that this type of density is equal to one or bounded below by a constant, this corresponds to the maximum clique and maximum quasiclique problem respectively, which are both NP-hard~\cite{pattillo2013maximum,karp1972reducibility}. 
There also exists a wealth of different dense subgraph optimization problems that are nontrivial yet polynomial time solvable. 
The most common is the densest subgraph problem, which seeks a subgraph maximizing the ratio between the number of induced edges and the number of nodes. 
Another common but seemingly quite different notion of density is to find a $k$-core of a graph~\cite{seidman1983kcore}, which is a maximal connected subgraph in which all induced node degrees are at least $k$. The maximum value of $k$ for which the $k$-core of a graph is non-empty is called the degeneracy of the graph. Throughout the text, we will use the term \emph{maxcore} to refer to the $k$-core of a graph when $k$ is the degeneracy.

Given the wide variety of existing applications for dense subgraph discovery, it is no surprise that a plethora of different combinatorial objective functions have been considered in practice beyond the objectives mentioned above~\cite{tsourakakis2013denser,tsourakakis2015k,sariyuce2015finding,sariyuce2018local,farago2008general,qin2015locally,tsourakakis2019novel,feige2001dense,lee2010survey,gionis2015dense}. However, as a result, it can be challenging to navigate the vast landscape of existing methods and evaluate tradeoffs between them in practice. In many cases, it is useful to uncover a range of different types of dense subgraphs in the same graph, but here again there is a challenge of knowing in which dimension or specific notion of density one would like to see variation. 

In light of these challenges, we define a simple, single-parameter family of objective functions for dense subgraph discovery. 
Our framework lets us (i) evaluate the tradeoffs between existing methods and (ii) uncover a hierarchy of dense subgraphs in the same graph, specifically in terms of the degree distribution of nodes in an induced subgraph. The family of objectives is motivated by a simple observation: the maxcore of a graph is the subgraph of maximum \emph{smallest} induced degree, while the standard densest subgraph is the one of maximum \emph{average} induced degree. Given that both of these are well-known and oft-used, 
it is natural to ask what other functions of induced node degrees are meaningful to optimize. 

To answer this, we introduce the \emph{generalized mean densest subgraph problem}, which is parameterized by a single parameter $p \in \mathbb{R}\cup \{-\infty, \infty\}$. We show that the maxcore and standard densest subgraph problems are obtained as special cases when $p = -\infty$ and $p = 1$, respectively.
Another existing notion of dense subgraphs that places a higher emphasis on large degrees is recovered when $p = 2$~\cite{farago2008general}.
And the limiting case of $p = \infty$ is solved by simply returning the entire graph.
Other values of $p$ lead to objectives that have not been considered previously, but correspond to natural alternative objectives for dense subgraph discovery. For example, in the case of $p = 0$, our problem is equivalent to finding a subgraph that maximizes the average logarithm of degree. Figure~\ref{fig:landscape} summarizes the objective function landscape captured by our framework.

\begin{figure}[t]
	\centering
	\includegraphics[width=\linewidth]{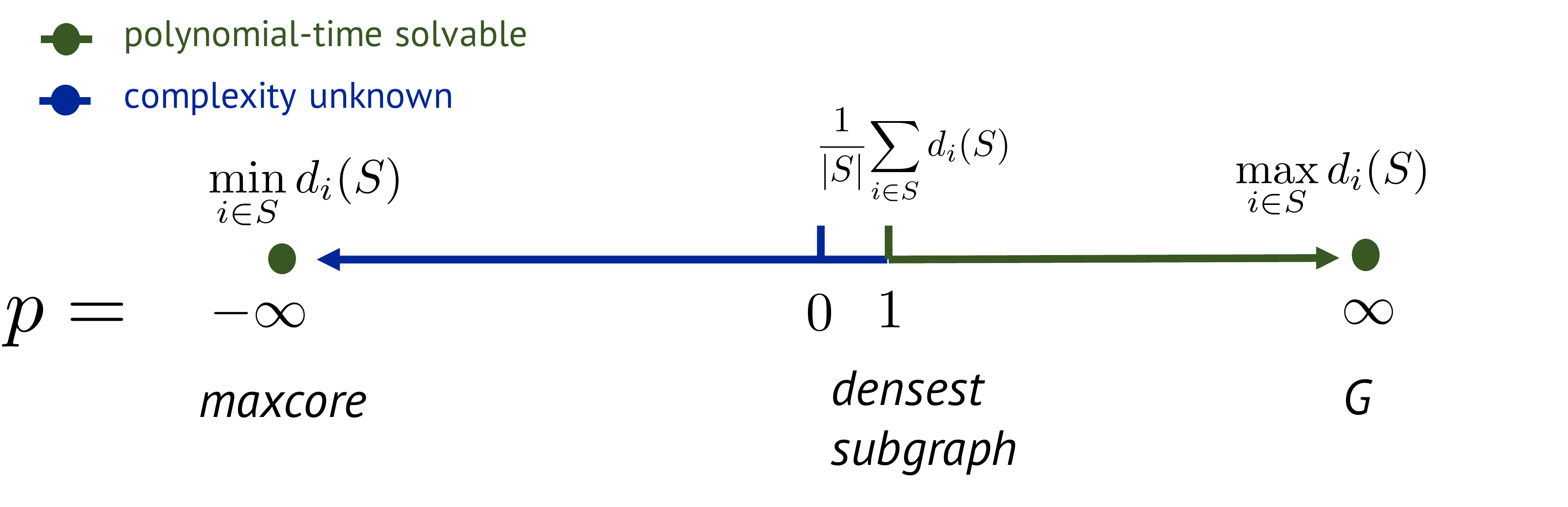}
	\caption{Summary of the objective function landscape and algorithmic results for our newly introduced generalized mean densest subgraph problem. The maxcore and densest subgraph problems are special cases of this simple parameterized objective. For $p\geq 1$, we give an optimal algorithm based on submodular minimization, and a fast greedy algorithm returning a $1/(p+1)^{1/p}$-approximation.
	}
	\label{fig:landscape}
\end{figure}

In addition to unifying a number of previous approaches, our framework is accompanied by several novel theoretical guarantees and algorithmic techniques for finding dense subgraphs. We first prove that our objective is polynomial-time solvable when $p \geq 1$, giving an algorithm based on repeated submodular minimization. 
We then present a variety of theoretical and empirical results on peeling algorithms for dense subgraph discovery in the context of our objective. These algorithms are approximation algorithms but are much faster than submodular minimization. They rely on repeatedly removing a single vertex at a time in order to shrink a graph down into a denser subgraph. The standard approach for peeling iteratively removes the node of smallest degree~\cite{asahiro2000greedily}. Surprisingly, we prove that despite seeming like a natural approach, this well-known standard peeling algorithm, which optimally solves the $p =-\infty$ objective~\cite{matula1983smallest} and provides a $1/2$-approximation for the $p = 1$ objective~\cite{charikar2000greedy,khuller2009finding}, can yield arbitrarily bad results when $p > 1$. 
However, we then design a more sophisticated but still fast peeling algorithm that is guaranteed to return a $1/(1+p)^{1/p}$ approximation for any $p \geq 1$. Although we prove that this bound is asymptotically tight, this approximation guarantee is always at least $1/2$, and converges to $1$ as $p \rightarrow \infty$. 
When $p = 1$, both the method and its guarantee reduce to that of the standard peeling algorithm.


 In order to greedily optimize the $p > 1$ version of our objective, our method must take into account not only how removing a node affects its own degree (i.e., it disappears), but also how this affects the node's neighbors in the graph and their contribution to the generalized objective. This gives our method a certain level of ``foresight'' when removing nodes, that is not present in the strategy of the standard peeling algorithm. 
 In light of this observation, our framework contains both theoretical and empirical examples where our new peeling method can outperform the standard peeling algorithm in finding dense subgraphs. For example, on many real-world graphs, we find that running our method with a value of $p$ slightly larger than 1 will typically produce sets with a better average degree than standard peeling, even though our method is technically greedily optimizing a different objective.
 
 We apply our methods on a range of different sized graphs from various domains, including social networks, road networks, citation networks, and web networks. We show that for small graphs with up to 1000 nodes, we can optimally find $p$-mean densest subgraphs using submodular minimization, and that on these graphs our approximation algorithm does a much better job approximating the optimal sets than its theory guarantees.  Our greedy peeling algorithm is also fast and scales to large datasets, and is able to uncover different meaningful notions of dense subgraphs in practice as we change our parameter $p$. 

In summary, we present the following contributions.
\begin{itemize}[topsep=0pt,leftmargin=*]
	\item We introduce the generalized mean densest subgraph problem and show that the well-studied maxcore and densest subgraph problems are special cases.
	\item We give a polynomial-time algorithm for optimally solving the objective for any $p \in [1, \infty]$, by showing that the decision version of the problem can be solved via submodular minimization.
	Existing linear-time algorithms solve the $p = -\infty$ case, which reduces to maxcore.
	\item We provide a faster greedy approximation algorithm that returns a $1/(p+1)^{1/p}$ approximation for $p \geq 1$. We show a class of graphs for which this approximation is tight. As $p \rightarrow \infty$, this approximation converges to 1. We also prove that for any $p > 1$, the greedy algorithm for the standard densest subgraph problem returns arbitrarily bad approximations on certain graph classes.
	\item We use our framework and methods to identify a range of different types of dense subgraphs and compare the performance of different peeling algorithms, on graphs coming from a wide variety of domains. 
\end{itemize}

%

\section{Technical Preliminaries}


Let $G = (V,E)$ be an undirected and unweighted graph. For $v \in V$, let $\mathcal{N}(v) = \{u \in V \colon (u,v) \in E\}$ denote the neighborhood of node $v$, and $d_v = |\mathcal{N}(v)|$ be its degree. For a set $S \subset V$, let $E_S$ denote the set of edges joining nodes in $S$ and $d_v(S) = |\mathcal{N}(v) \cap S|$ be the degree of $v$ in the subgraph induced by $S$. If $v \notin S$, then $d_v(S) = 0$. 

\paragraph{Dense Subgraph Problems}
The densest subgraph problem seeks a set of nodes $S$ that maximizes the ratio between the number of edges and nodes in the subgraph induced by $S$:
\begin{equation}
\label{density}
\maximize_{S \subseteq V} \frac{|E_S|}{|S|}.
\end{equation} 
 The numerator of this objective is equal to half the sum of induced node degrees in $S$, and therefore this problem is equivalent to finding the maximum average degree subgraph:
\begin{equation}
\label{maxavgdegree}
\max_{S \subseteq V} \frac{\sum_{v \in S} d_v(S)}{|S|} = \max_{S \subseteq V} \textbf{avg}_{v \in S}\, d_{v}(S).
\end{equation}
This problem is known to have a polynomial-time solution~\cite{goldberg1984finding,gallo1989fast}, as well as a fast greedy peeling algorithm that is guaranteed to return a $1/2$-approximation~\cite{khuller2009finding,charikar2000greedy}. 

Another well-studied dense subgraph problem is to find the $k$-core of a graph for a positive integer $k$~\cite{dorogovtsev2006k,carmi2007model,shin2016corescope,malliaros2020core}. 
The $k$-core of a graph is a maximal connected subgraph in which all nodes have degree at least $k$. The maximum value of $k$ for which the $k$-core of a graph is non-empty is called the \emph{degeneracy} of the graph. If $k^*$ is the degeneracy of $G$, we will refer to the $k^*$-core of $G$ as the \emph{maxcore}. Finding the maxcore is equivalent to finding the subset $S$ that maximizes the minimum induced degree:
\begin{equation}
\label{maxcore}
\max_{S \subseteq V}  \min_{v \in S} \, d_v(S).
\end{equation}
For any value of $k$, the $k$-core of a graph can be found in linear time via a greedy peeling algorithm
that repeatedly removes nodes with degree less than $k$~\cite{matula1983smallest}.

\paragraph{Generalized Means}
For a vector of positive real numbers $\textbf{x} = \begin{bmatrix} x_1 & x_2 & \cdots & x_n \end{bmatrix} \in \mathbb{R}^n_+$, the generalized mean (power mean or $p$-mean) with exponent $p$ of $\textbf{x}$ is defined to be
\begin{equation}
\label{genmean}
 M_p(\vx) =\left( \frac{1}{n} \sum_{i = 1}^n (x_i)^p\right)^{1/p}.
\end{equation}
For $p \in \{\infty, -\infty, 0\}$, the mean can be defined by taking limits, so that $M_{\infty}(\vx) = \min_i x_i$, $M_{-\infty} = \max_i x_i$, and $M_0(\vx) = (x_1x_2\cdots x_n)^{1/n}$ (the $p = 0$ case is often called the geometric mean).

\paragraph{Submodularity and the power function}
We review a few useful preliminaries on submodular functions and properties of the power function 
$\phi(x) = \vert x \rvert^p$ for $p \geq 0$ that will be useful for our algorithmic results. 
For a discrete set $\Omega$, a function $f \colon \Omega \rightarrow \mathbb{R}$ is submodular if for any subsets $A, B \subseteq \Omega$ it satisfies
\begin{equation}
\label{submod}
f(A) + f(B) \geq f(A\cap B) + f(A \cup B).
\end{equation}
If the above inequality is reversed, $f$ is referred to as a \emph{supermodular} function, and if we replace it with equality, the function is called \emph{modular}. We have the following two simple observations about properties of the power function,
which follow from the fact that $\phi$ is concave when $p \in [0,1]$, and convex when $p \geq 1$. 
\begin{observation}
	\label{obs:g}
	The function $g(A) = |A|^p$ for $A \in \Omega$ is submodular when $p \in [0,1]$ and supermodular when $p \geq 1$. 
\end{observation}
\begin{observation}
	\label{obs:power}
	Let $p \geq 1$ and $x \geq 1$. Then $$p(x-1)^{p-1} \leq x^p - (x-1)^p \leq px^{p-1}.$$
\end{observation}


\section{The $p$-Mean Densest Subgraph}
We now introduce a new single-parameter family of dense subgraph objectives based on generalized means of degree sequences. 
Abusing notation slightly, we define the $p$-density of $S \subseteq V$ to be
\begin{equation}
\label{subgraphpmean}
{ M_p(S) = \left( \frac{1}{|S|} \sum_{v \in S} [d_v(S)]^p\right)^{1/p}.}
\end{equation}
The \emph{generalized mean densest subgraph problem} is then to find a set of node $S$ that maximizes $M_p(S)$. We  also refer to this as the $p$-mean densest subgraph problem to make the parameter $p$ explicit. It is worthwhile to note that this objective increases monotonically with $p$, ranging from the minimum degree of $S$ when $p = -\infty$ to the maximum degree of $S$ when $p = \infty$.
%

\subsection{Equivalence Results and Special Cases}
The definition of generalized mean, and the characterizations of the densest subgraph and max-core problem given in~\eqref{maxavgdegree} and~\eqref{maxcore}, immediately imply the following equivalence results.
\begin{lemma}
The standard densest subgraph problem~\eqref{maxavgdegree} is equivalent to the $1$-mean densest subgraph problem. The maxcore objective is equivalent to the $-\infty$-mean densest subgraph problem. 
\end{lemma}	
Although studied less extensively, an equivalent variant of the $2$-mean densest subgraph has also been previously considered~\cite{farago2008general} as a way to find dense subgraphs that place a higher emphasis on the largest degrees. Our new objective therefore unifies existing dense subgraph problems, and suggests a natural way to interpolate them as well as find new meaningful notions of dense subgraphs. 

For finite $p > 0$, maximizing $M_p(S)$ is equivalent to maximizing $[M_p(S)]^p$. In other words, the $p$-mean densest subgraph problem seeks a subgraph with the highest average $p$th-power degree:
\begin{equation}
\label{fp}
{f_p(S) = \sum_{v \in S} \frac{d_v(S)^p}{|S|}.}
\end{equation}
For $p = \infty$, objective~\eqref{fp} is no longer meaningful, but from the standard definition of $\infty$-mean we know that $M_\infty(S) = \lim_{p\rightarrow \infty} M_p(S)$ is simply the maximum degree of the induced subgraph. 
This makes intuitive sense given that objective~\eqref{fp} places higher and higher emphasis on large degrees as $p \rightarrow \infty$. The $p = \infty$ objective is trivially solved by taking the entire graph $G$, though this will not necessarily be optimal for large but finite values of $p$. 

When $p = 0$, the generalized mean coincides with the geometric mean, which for our dense subgraph framework means that
\begin{equation}
\label{eq:p0}
  M_{0}(S) = \left( \prod_{v \in S} d_v(S)  \right)^{1/|S|}.
\end{equation}
As long as $S$ does not include zero degree nodes, maximizing this quantity is equivalent to maximizing $\log M_{0}(S)$, which amounts to finding the subgraph with maximum average log-degree:
\begin{equation}
\label{eq:logdegree}
{  f_0(S) = \log M_{0}(S)= \sum_{v \in S} \frac{\log d_v(S)}{|S|}.}
\end{equation}
While natural, this objective has not been previously considered for dense subgraph problems.

Finally, for finite $p < 0$, maximizing $M_p(S)$ is equivalent to maximizing $(M_p(S))^{-p}$, which equals
\begin{equation}
\label{eq:pnegative}
(M_p(S))^{-p} = \left( \frac{\sum_{v \in S} [d_v(S)]^p}{|S|} \right)^{-1} = \left( \text{avg}_S  [d_v(S)]^p\right)^{-1}.
\end{equation}
In other words, while the $p = 1$ objective maximizes the average degree, the $p = -1$ objective seeks the to maximize one over the average inverse degree (corresponding to the harmonic mean of the degrees). 
For other $p < 0$ the goal is to maximize one over the average $p$th power of the degrees. 

\subsection{Comparison with Existing Objectives}
Before moving on we compare and contrast our $p$-mean objective against similar generalizations of the densest subgraph problem. 
See Section~\ref{sec:related} for more related work.  

\textbf{The $k$-clique densest subgraph.} Tsourakakis~\cite{tsourakakis2015k} previously introduced the $k$-clique densest subgraph problem ($k$-DS), which seeks a set of nodes $S$ that minimizes the ratio between the number of $k$-cliques and the number of nodes in $S$. The standard densest subgraph problem is recovered when $k =2$. The author showed that this objective can be maximized in polynomial time for any fixed value of $k$, and also gave a $1/k$-approximate peeling algorithm by generalizing the standard greedy peeling algorithm~\cite{charikar2000greedy,asahiro2000greedily}. In practice, even when $k=3$, maximizing this objective can produce subgraphs that are much closer to being near cliques than the standard densest subgraph solution.

\textbf{$\textbf{F}$-density.} An even more general objective called $\textbf{F}$-density was introduced by Farag\'{o}~\cite{farago2008general}. Given a family of graphs $\textbf{F}$, this objective seeks a set $S$ that maximizes the ratio between the number of instances of $\textbf{F}$-graphs in $S$, and $|S|$. The $k$-DS problem is recovered when $\textbf{F}$ includes only the clique on $k$ nodes. Interestingly, Farag\'{o}~\cite{farago2008general} showed that when $\textbf{F} = \{P_2, P_3\}$, where $P_i$ is the path graph on $i$ nodes, then the $\textbf{F}$-density is the following ratio:
\begin{equation}
\label{farago2}
{  \frac{1}{2} \frac{\sum_{v \in S} d_v(S)^2}{|S|}.}
\end{equation}
The maximizers of this objective are the same as those for the $2$-mean densest subgraph problem in our framework, 
although there are differences in terms of approximation guarantees of algorithms.

\textbf{Discounted average degree and $\textbf{f}$-densest subgraph.} The above objectives generalize the standard densest subgraph problem by changing the \emph{numerator} of the objective. Generalizing in a different direction, Kawase and Miyauchi~\cite{kawase2018} introduced the $f$-densest subgraph problem, which seeks a set $S \subseteq V$ maximizing $|E_S|/f(|S|)$ for a convex or concave function $f$. This includes the special case $f(S) = |S|^{\alpha}$ for $\alpha > 0$, which generalizes the earlier notion of \emph{discounted average degree} considered by Yanagisawa and Hara~\cite{yanagisawa2018}. When $f$ is concave, maximizing the objective will always produce output sets that are larger than or equal to optimizers for the densest subgraph problem. Convex $f$ produces outputs sets that are always smaller than or equal to densest subgraph solutions~\cite{kawase2018}.

\textbf{Comparison with $p$-mean densest subgraph.}
The $p$-mean densest subgraph problem is similar to the above objectives in that they are all parameterized generalizations of the densest subgraph problem. However, each generalizes the objective in a different direction.
The $k$-clique densest subgraph problem parameterizes preferences for obtaining clique-like subgraphs, while $\textbf{F}$-density more generally controls the search for dense subgraphs that are rich in terms of specified subgraph patterns (graphlets). The discounted average degree and $f$-densest subgraph objectives reward set sizes differently, while keeping the same emphasis on counting edges in the induced subgraph.
Meanwhile, our objective encapsulates different preferences in terms of induced node degrees. Many subcases amount to maximizing averages of different functions on node degrees.
Importantly, the max-core problem is not captured as a special case of $\textbf{F}$-density, $k$-DS, or the $f$-densest subgraph. 

Our objective and the $f$-densest subgraph problem are tangentially similar in that they both can involve the power function $f(x) = |x|^\alpha$ in some way. However, a major difference is that the latter objective shares the same numerator (the induced edge count $|E_S|$) as the standard densest subgraph, which is not the case for our problem. This leads to substantial differences in terms of the computational complexity and the output sets that maximize these objectives. For example, the $f$-densest subgraph with $f(x) = |x|^\alpha$ is NP-hard when $\alpha > 1$ and is guaranteed produce smaller output sets than the standard densest subgraph in this case, and is also known to have trivial constant-sized optimal solutions when $\alpha > 2$~\cite{kawase2018}. In contrast, $p \geq 1$ is in fact the easier regime for our problem and will often, but not always, produce larger output sets. 
Finally, we note that the relationship between $\textbf{F}$-density and $2$-mean density does not appear to generalize to other values of $p$, even if we restrict to positive integer values of $p$.
In particular, by checking a few small examples, it is easy to confirm that the $3$-mean density objective is not equivalent to $\textbf{F}$-density when $\textbf{F} = \{P_2, P_3, P_4\}$.

\section{Algorithms}
We now present new algorithms for the $p$-mean densest subgraph problem. We show that the objective can be solved in polynomial time for any $p \geq 1$ via submodular minimization. For the same parameter regime, we show that the standard greedy peeling algorithm~\cite{asahiro2000greedily,charikar2000greedy,khuller2009finding} can return arbitrarily bad results, but a more sophisticated generalization yields a $(1+p)^{1/p}$ approximation, which is also tight. Algorithms for $p \in (-\infty, 1)$ appear more challenging---both our submodular minimization technique and greedy approximation do not hold in this case. Improved algorithms or hardness results for this regime are a compelling avenue for future research.

\subsection{An Optimal Algorithm for $p \geq 1$}
For finite $p > 0$, $\argmax M_p(S) = \argmax f_p(S)$, so for simplicity we focus on the latter objective. In the next section we will use the fact that a $C$-approximate solution for $f_p(S)$ provides a $C^{1/p}$ approximate solution for $M_{p}(S)$. 

Given a fixed $\alpha > 0$, the decision version of our problem asks whether there exists some $S$ such that $f_p(S) \geq \alpha$, or equivalently
%
$\sum_{i \in S} d_i(S)^p -\alpha |S| \geq 0$.
We can obtain a yes or no answer to this question by solving the following optimization problem
\begin{equation*}
 \max_{S \subseteq V} \,\, \psi(S) = \sum_{i \in S} d_i(S)^p -\alpha |S|.
\end{equation*}
When $p = 1$, it is well known that this can be solved by solving a minimum $s$-$t$ cut problem~\cite{goldberg1984finding}, which is itself a special case of submodular minimization (equivalently, supermodular function maximization). The following result guarantees that we can still get a polynomial time algorithm for the $p$-mean densest subgraph when $p\geq 1$ by using general submodular minimization.
\begin{lemma}
	\label{lem:pgeq1}
	If $p \geq 1$ and $\alpha > 0$, the function $\psi(S) = \sum_{i \in S} d_i(S)^p -\alpha |S|$ is supermodular. 
\end{lemma}
\begin{proof}
	The function $g(S) = |S|$ is modular,
	so it suffices to show that the function $h(S) = \sum_{i \in S} d_i(S)^p$ is supermodular whenever $p \geq 1$.
	For a node $v \in V$ and an arbitrary set $R \subset V$, if $v \notin R$ then define $d_v(R) = 0$. This allows us to write $h(R)$ as a summation over all nodes in $V$ without loss of generality:
	\begin{equation*}
	  h(R) = \sum_{v \in V} d_v(R)^p.
	\end{equation*}
	
	We would like to show that for arbitrary sets $S$ and $T$, we have
	\begin{equation*}
	  \sum_{v \in V} d_v(S)^p + d_v(T)^p \leq \sum_{v \in V} d_v(S\cap T)^p + d_v(S \cup T)^p,
	\end{equation*}
	which is true if we can show that for every $v \in V$ and $p\geq 1$,
	\begin{equation}
	\label{dp}
	d_v(S)^p + d_v(T)^p \leq d_v(S\cap T)^p + d_v(S \cup T)^p.
	\end{equation}
	To prove~\eqref{dp}, let $E_v$ be the set of edges adjacent to $v$, and note that all of the terms in inequality~\eqref{dp} represent the number of edges in a certain subset of $E_v$. Specifically, define
	\begin{align*}
	&A = \{ (i,j) \in E_v \colon i \in S, j \in S\},\quad
	B = \{ (i,j) \in E_v \colon i \in T, j \in T\}, \\
	&C = \{ (i,j) \in E_v \colon i \in S\cap T, j \in S\cap T\},\\
	&D = \{ (i,j) \in E_v \colon i \in S\cup T, j \in S\cup T\} .
	\end{align*}
	Consider the sets obtained by unions and intersections of $A$ and $B$:
	\begin{align*}
	A\cap B &= \{ (i,j) \in E_v \colon i \in S\cap T, j \in S\cap T\}  = C \\
	A \cup B &= \{ (i,j) \in E_v \colon i,j \in S \text { OR }i,j \in T \} \subseteq D .
	\end{align*}
	These sets are related to sets of edges that contribute to degree counts for node $v \in V$:
	\begin{align*}
	&d_v(S) = |A|,\quad
	d_v(T) = |B|,\quad
	d_v(S\cap T) = |C| = |A\cap B|, \\
	&d_v(S \cup T)= |D| \geq |A \cup B|.
	\end{align*}
	Recall from Observation~\ref{obs:g} that if $\Omega$ is a discrete set, 
	then the function $z(A) = |A|^p$ for $A \subseteq \Omega$ is supermodular whenever $p \geq 1$. 
	Therefore, for these sets of edges,
	\begin{align*}
	&d_v(S\cap T)^p + d_v(S \cup T)^p \\
	&= |C|^p + |D|^p
	\geq |A\cap B|^p + |A\cup B|^p
	\geq  |A|^p + |B|^p = d_v(S)^p + d_v(T)^p.
	\end{align*}
	So the desired inequality is shown.
\end{proof}
Using submodular minimization algorithms as a black box~\cite{Orlin2009}, we can perform binary search on $\alpha$ to find the maximum value of $\alpha$ such that $\phi(S) \geq 0$, which is equivalent to saying $f_p(S) \geq \alpha$. The number of binary search steps necessary to exactly optimize $f_p(S)$ can be easily bounded by a polynomial in $n$. The number of different values that the numerator $\sum_{v \in S} d_v(S)^p$ can take on is trivially bounded above by $2^n$, as this is the number of distinct ways to bipartition the graph. The denominator can take on $n$ values. We can use $0$ and $\sum_{v \in V} d_v^p$ as bounds for our binary search, and even for an extremely pessimistic case, $O(n)$ binary search steps would be necessary, though typically it will be far less in practice. We conclude the following result.
\begin{theorem}
	For any graph $G = (V,E)$ and $p \geq 1$, the $p$-mean densest subgraph can be found in polynomial time.
\end{theorem}
This result is intended to serve mainly as a theoretical result confirming the polynomial-time solvability of the problem for $p \geq 1$. 
We next turn to more practical greedy approximation algorithms. 

\subsection{Failure of the Standard Peeling Algorithm}
The standard peeling algorithm for both the maxcore and densest subgraph problem is to start with the entire graph $G$ and repeatedly remove the minimum degree node until no more nodes remain. We will refer to this algorithm generically as \textsc{SimplePeel}. This algorithm produces a set of $n$ subgraphs $S_1, S_2, \hdots, S_n$, one of which is guaranteed to solve the maxcore problem~\cite{matula1983smallest}, and another of which is guaranteed to provide at least a $1/2$-approximation to the standard densest subgraph problem~\cite{charikar2000greedy}. Trivially, this method also yields and optimal solution for $p = \infty$, since the entire graph $G$ solves the $\infty$-mean objective.

Given the success of this procedure for $p\in  \{-\infty, 1, \infty\}$, it is natural to wonder whether it can be used to obtain optimal or near optimal solutions for other values of $p$. Focusing on the $p \geq 0$ case, we know that if $d_v = \min_{i \in S} d_i(S)$ for a subset $S$, then it is also true that $d_v = \min_{i \in S} d_i(S)^p$, again suggesting that this strategy might be effective. 
However, surprisingly, we are able to show that the simple peeling algorithm can perform arbitrarily poorly for any $p > 1$. 
To show this, we consider the performance of the algorithm on the same class of graphs that has been used to show that the $1/2$-approximation for the standard peeling algorithm is tight for the 1-mean objective~\cite{khuller2009finding}.
\begin{lemma}
	\label{lem:failure}
	Let $p > 1$ and $\varepsilon \in (0,1)$ be fixed constants. There exists a graph $G$ such that applying \textsc{SimplePeel} on $G$ will yield an approximation worse than $\varepsilon$ for the $p$th power degree objective.
\end{lemma}
\begin{proof}
	The proof is by construction. Let $G_1 = (V_1, E_1)$ be a complete bipartite graph with $D$ nodes on one side and $d$ nodes on the other, and $G_2 =  (V_2, E_2)$ be a disjoint union of $D$ cliques of size $d+2$. We treat $d$ as a fixed constant and $D$ as a value that will grow without bound.
	Define graph $G = (V_1 \cup V_2, E_1 \cup E_2)$ to be the union of $G_1$ and $G_2$. We have $f_p(V_2) = (d+1)^p$, and 
	\begin{equation*}
	f_p(V_1)  = \frac{Dd^p + dD^p}{d + D}.
	\end{equation*}
	As $D\rightarrow \infty$, $V_1$ is the best $p$-mean densest subgraph, with$f_p(V_1)$ behaving as $dD^{p-1}$. 
	However, \textsc{SimplePeel} will start by removing all of the nodes in $V_1$, since on one side of the bipartite graph, all nodes have degree $d$,
	whereas all nodes $V_2$ have degree $d + 1$. 
	At the outset of the algorithm, we start with all of $V$ and note that
	\begin{equation*}
	f_p(V) = \frac{Dd^p + dD^p + D(d+2)(d+1)^p}{D+ d + D(d+2)} < 
	d^{p-1} + D^{p-1} + (d+1)^{p}.
	\end{equation*}
	We see that as $D \rightarrow \infty$, this provides only a $1/d$ approximation to $f_p(V_1)$.	As \textsc{SimplePeel} removes nodes from $V_1$, the approximation gets worse, until we are left with a subgraph with maximum average $p$th power degree of $(d+1)^p$. Since $d$ was an arbitrary constant, we can choose $d = \ceil{2/\varepsilon}$. Then, for large enough $D$ we will have 
	\begin{align*}
	\frac{1}{d} < \frac{f_p(V)}{f_p(V_1)} < \frac{2}{d} < \varepsilon.
	\end{align*}
	Therefore, the approximation is worse than $\varepsilon$.
\end{proof}
The above result means that for any $\varepsilon \in (0,1)$ and $p > 1$ fixed, there exists a graph with an optimal $p$-mean densest subgraph $S^*$, such that the simple greedy method will return a set $\hat{S}$ satisfying
\begin{equation*}
f_p(\hat{S}) < \varepsilon f_p(S^*) \implies M_p(\hat{S}) < \varepsilon^{1/p} \max_{S \subseteq V} M_p(S).
\end{equation*}
Since $p$ is fixed and $\varepsilon$ is arbitrarily small, this means the simple greedy method can do arbitrarily badly when trying to approximate the $p$-mean densest subgraph problem.

\subsection{Generalized Peeling Algorithm when $p \geq 1$}
\begin{algorithm}[t]
	\caption{Generalized Peeling Algorithm (\textsc{GenPeel}-$p$)}
	\label{alg:genpeel}
	\begin{algorithmic}
		\State \textbf{Input}: $G = (V,E)$, parameter $p \geq 1$
		\State \textbf{Output}: Set $S^* \subseteq V$, satisfying $f_1(S^*) \geq \frac{1}{p+1} \max_S f_1(S)$.
		\State $S_0 \leftarrow V$
		\For{$i = 1$ to $n$} 
		\State $\ell = \argmin_j \Delta(d_j(S_{i-1}))$
		\State $S_i = S_{i-1}\backslash\{\ell\}$
		\EndFor
		\State Return $\max_i f_p(S_i)$.
	\end{algorithmic}
\end{algorithm}
The failure of \textsc{SimplePeel} can be explained by noting that when $p > 1$, removing a minimum degree node from a subgraph $S$ does not in fact \emph{greedily} improve $p$-density. Consider a node set $S$ and its average $p$th power degree function $f_p(S)$. Removing any $v \in S$ will change the denominator of $f_p(S)$ in exactly the same way. Therefore, to greedily improve the objective by removing a single node, we should choose the node that leads to the minimum decrease in the numerator $\sum_{v \in S} d_v(S)^p$. Importantly, it is not necessarily the case that $j = \argmin_{v \in S} d_v(S)^p$ is the best node to remove. This would account only for the fact that the numerator decreases by $d_j(S)^p$, but ignores the fact that removing $j$ will also affect the degree of all other nodes that neighbor $j$ in $S$. For example, if $j$ has a small degree but neighbors a high degree node $u$, then when $p > 1$, removing node $j$ and decreasing $u$'s degree could substantially impact the $p$-density objective.
For a graph $G$, node set $S$, and arbitrary node $j \in S$, the following function reports the exact decrease in the numerator of $f_p(S)$ resulting from removing $j$:
\begin{equation}
\label{Delta}
\Delta_j(S) = d_j(S)^p + \sum_{i \in \mathcal{N}(j) \cap S} d_{i}(S)^p - [d_i(S) - 1]^p.
\end{equation}
In other words, note that $\sum_{i \in S} d_i(S)^p$ is the numerator before removing $j$, and after removing it we have a new numerator 
\begin{equation*}
\sum_{i \in S\backslash\{j\}} d_i(S\backslash\{j\})^p = \sum_{i \in S} d_i(S)^p - \Delta_j(S).
\end{equation*}
Observe that when $p = 1$, $\Delta_j(S) = 2d_j(S)$, which explains why it suffices to remove the minimum degree node in order to greedily optimize the $p = 1$ variant of the objective. Based on these observations, we define a generalized peeling algorithm, which we call \textsc{GenPeel}-$p$, or simply \textsc{GenPeel} when $p$ is clear from context, based on iteratively removing nodes that minimize~\eqref{Delta}. Pseudocode is given in Algorithm~\ref{alg:genpeel}. We prove the following result.
\begin{theorem}
	Let $G = (V,E)$ be a graph, $p \geq 1$, and $T$ be the $p$-mean densest subgraph of $G$. Then \textsc{GenPeel}-$p$ returns a subgraph ${S}$ satisfying $(p+1) f_p({S})  \geq f_p(T)$, i.e., $(p+1)^{1/p}M_p({S}) \geq  M_p(T)$.
\end{theorem}
\begin{proof}
 Define $\gamma = f_p(T)$, which implies that $\sum_{i \in S} d_i(S)^p - |T| \gamma  = 0$.
	Since $T$ is optimal, removing a node $j$ will produce a set with $p$-density at most $\gamma$, and therefore, we have
	\begin{align*}
	&\frac{\sum_{i \in T} d_i(T)^p - \Delta_j(T)}{|T| - 1} \leq \gamma \\
	\implies & \sum_{i \in T} d_i(T)^p  -\Delta_j(T) \leq |T| \gamma - \gamma \implies \gamma \leq \Delta_j(T).
	\end{align*}
	
	Observe that for any set $R \supseteq T$ and $j \in T$ we have $\Delta_j(T) \leq \Delta_j(R)$. This holds because $d_j(R)$ is larger than $d_j(T)$ when we grow the subgraph, and also because the function $d_{i}(R)^p - [d_i(R) - 1]^p$ also monotonically grows as $R$ gets bigger, \emph{as long as $p \geq 1$}. 
	
	Let $S$ denote the set maintained by the greedy algorithm right before the first node $j \in T$ is removed by peeling. Since $j$ is the first node to be removed by the peeling algorithm, we know that $S \supseteq T$ and $\Delta_j(T) \leq \Delta_j(S)$. Because $j = \argmin_{i \in S}\Delta_i(S)$, we know $\Delta_j(S)$ is smaller than the average value of $\Delta_i(S)$ across nodes in $S$, so
	\begin{align*}
	\gamma \leq &\Delta_j(T) \leq \Delta_j(S) \leq \frac{1}{|S|} \sum_{\ell \in S} \Delta_\ell(S) \\
	&= \frac{1}{|S|} \left( \sum_{\ell \in S} d_\ell(S)^p + \sum_{\ell \in S} \sum_{i \in \mathcal{N}(\ell) \cap S} d_i(S)^p - [d_i(S)-1]^p\right)\\
	&\leq \frac{1}{|S|} \left( \sum_{\ell \in S} d_\ell(S)^p + \sum_{\ell \in S} \sum_{i \in \mathcal{N}(\ell) \cap S} pd_i(S)^{p-1} \right).
	\end{align*}
	The last step follows from the bound in Observation~\ref{obs:power}. For every $i \in S$, the value $pd_i(S)^{p-1}$ show up exactly $d_i(S)$ times in the double summation in  the second term---once for every $\ell \in S$ such that $i$ neighbors $\ell$ in $S$. Therefore:
	\begin{equation*}
	{ \gamma \leq \frac{ \sum_{\ell \in S} d_\ell(S)^p}{|S|} + \frac{p \sum_{i \in S} d_i(S)^{p}}{|S|} = (p+1)f_p(S).}
	\end{equation*} 
\end{proof}
In the appendix, we provide an in-depth graph construction to show that this approximation guarantee is asymptotically tight for all values of $p$. Nevertheless, for $p \geq 1$, the quantity $(p+1)^{1/p}$ decreases monotonically and has a limit of one, meaning that the $p$-density problem in fact becomes easier to approximate as we increase $p$. In fact, the $1/2$-approximation for the standard densest subgraph problem is the worst approximation that this method obtains for any $p \geq 1$. The fact that the approximation factor converges to 1 also intuitively matches the fact that when $p = \infty$, the optimal solution is trivial to obtain by keeping all of $G$. 

\emph{Key differences between peeling algorithms.}
Before moving on, it is worth noting two key differences about removing nodes based on degree (\textsc{SimplePeel}) and removing nodes based on $\Delta_j$ (\textsc{GenPeel}). The first is purely practical: it is faster to keep track of node degrees than to keep track of changes to $\Delta_j$ for each node $j$ when peeling. Removing a node $v$ will change the degrees of nodes within a one-hop neighborhood, but will change $\Delta_j$ values for every node within a two-hop neighborhood of $v$. This additional detail complicates the runtime analysis and implementation of \textsc{GenPeel}. For $m = |E|$, a naive runtime bound of $O(nm)$ can still be obtained by first realizing that as long as degrees are known, the value of $\Delta_j$ can be computed in $O(d_j)$ time, and so all $\Delta_j$ values can be computed in $O(\sum_j d_j) = O(m)$ time. There are $n$ rounds overall, and in each round it takes $O(n)$ time to find the minimizer of $\Delta_j$, plus $O(n)$ time to update degrees after a node is removed, plus $O(m)$ time to update $\Delta_j$ values in preparation for the next round. 

Although this runtime bound is much worse than the linear time guarantee for the best implementation of \textsc{SimplePeel}, it is quite pessimistic, and in practice we can still scale up \textsc{GenPeel} to large datasets in practice. Furthermore, although \textsc{GenPeel} is slower, this difference in strategy is not without its advantages. The $1/2$-approximation for \textsc{SimplePeel} is tight on certain graph classes made up of disjoint cliques and complete bipartite graphs~\cite{khuller2009finding}. This is also the same class of graphs we used to show that the method does arbitrarily badly when $p > 1$. \textsc{SimplePeel} fails to find the best subgraph on these instances because it considers the degree of a node without sufficiently considering the effect on its neighborhood. Meanwhile, \textsc{GenPeel} finds the optimal solution on these graph classes for all $p \geq 1$. In our next section, we will demonstrate empirically that running \textsc{GenPeel}  with $p > 1$ actually outperforms \textsc{SimplePeel} in optimizing the $p = 1$ objective.


%

\section{Experiments}
\label{sec:experiments}
We now consider how our methods enable us to find a range of different meaningful notions of dense subgraphs in practice. 
We begin by showing that \textsc{GenPeel} does an excellent job of approximating the optimal $p$-mean densest subgraph problem, scales to large datasets, and detects subgraphs with a range of different meaningful notions of density. Using \textsc{GenPeel}-$p_0$ with $p_0 > 1$ can even be used to solve the $p = 1$ objective (standard densest subgraph) better than \textsc{SimplePeel} (which is equivalent to \textsc{GenPeel}-$1$) in many cases, even though \textsc{SimplePeel} greedily optimizes the $p = 1$ objective but \textsc{GenPeel}-$p_0$ with $p_0 > 1$ does not. All of our experiments were performed on a laptop with 8GB of RAM and 2.2 GHz Intel Core i7 processor. 
We use public datasets from the SNAP repository~\cite{snapnets} and the SuiteSparse Matrix collection~\cite{davis2011uf}. 
We implement~\textsc{GenPeel} and~\textsc{SimplePeel} in Julia, both using a simple min-heap data structure for removing nodes, which works well in practice. 
We implement our optimal submodular minimization approach in MATLAB, in order to use existing submodular optimization software~\cite{krause2010sfo}. All algorithm implementations and code for our experiments are available at~\url{https://github.com/nveldt/GenMeanDSG}.

\subsection{\textsc{GenPeel} Approximation Performance}
\begin{table*}[t!]
	\caption{
		We compare our peeling algorithm for the $p$-mean densest subgraph against \textsc{SimplePeel} (the $p = 1$ special case) and the maxcore of a graph (p = $-\infty$) in terms of various measures of density on a range of graphs. With some exceptions, increasing $p$ tends to produce larger sets with lower edge density. The $p = 1$ and maxcore case have the same runtime as they rely on finding the same ordering of nodes, with different stopping points. Simple peeling is faster, but our approach is still fast, and for max degree, averaged squared degree, and average degree, one of our new approaches leads to the best results in all cases. Our new methods in several cases also lead to the best results for edge density. We highlight in bold the best result obtained for these four different notions of density. The fact that \textsc{GenPeel}-$p_0$ with $p_0 > 1$ outperforms \textsc{SimplePeel} in terms of average degree is especially significant, since the latter is designed to greedily optimize average degree, the $p = 1$ objective.
	}
	\label{tab:snap}
	\centering
	\scalebox{0.95}{\begin{tabular}{l   |  l  | l l l     l l l  l l l l }
			\toprule
			&& \textbf{Astro} & \textbf{CM05} & \textbf{BrKite} & \textbf{Enron} & \textbf{roadCA} & \textbf{roadTX} & \textbf{webG} & \textbf{webBS} & \textbf{Amaz} & \textbf{YTube} \\
			&$|V|$& 17,903 &36,458 & 58,228 & 36,692 & 1,971,281 & 1,393,383 & 916,428 & 685,230 & 334,863 & 1,134,890  \\
			\textbf{Metric}& $|E|$ & 196,972 &171,734 &21,4078 &183,831 &2,766,607 &1,921,660 &4,322,051 &6,649,470 &925,872 &2,987,624 \\
			\midrule
			Size & \emph{maxcore}& 57& 30& 154& 275& 4568& 1579& 48& 392& 497& 845 \\
			& $p = 0.5$& 165& 30& 200& 469& 4568& 1579& 229& 392& 497& 1616 \\
			$|S|$& $p = 1.0$& 1151& 563& 219& 548& 11& 3721& 240& 392& 34& 1863 \\
			& $p = 1.05$& 1317& 565& 220& 556& 6& 26& 240& 392& 3848& 1900 \\
			& $p = 1.5$& 1564& 742& 242& 713& 12& 91& 243& 5352& 118& 3085 \\
			& $p = 2.0$& 1491& 938& 273& 1036& 185& 13& 4429& 34944& 550& 29639 \\
			\midrule
			Edge  & \emph{maxcore}& \textbf{1.0}& \textbf{1.0}& \textbf{0.502}& \textbf{0.256}& 0.001& 0.002& \textbf{0.994}& \textbf{0.529}& 0.014& \textbf{0.102} \\
			Density& $p = 0.5$& 0.348& \textbf{1.0}& 0.407& 0.159& 0.001& 0.002& 0.235& \textbf{0.529}& 0.014& 0.056 \\
			& $p = 1.0$& 0.052& 0.056& 0.372& 0.137& 0.345& 0.001& 0.227& \textbf{0.529}&\textbf{0.23}& 0.049 \\
			$|E_S|/{|S| \choose 2}$& $p = 1.05$& 0.045& 0.056& 0.37& 0.135& \textbf{0.733}& 0.166& 0.227& \textbf{0.529}& 0.002& 0.048 \\
			& $p = 1.5$& 0.039& 0.043& 0.335& 0.104& 0.333& 0.044& 0.225& 0.031& 0.082& 0.029 \\
			& $p = 2.0$& 0.041& 0.032& 0.29& 0.068& 0.02& \textbf{0.295}& 0.005& 0.001& 0.005& 0.001 \\
			\midrule
			Avg  & \emph{maxcore}& 56.0& 29.0& 76.87& 70.06& 3.32& 3.34& 46.71& \textbf{206.81}& 6.77& 86.07 \\
			Degree& $p = 0.5$& 57.02& 29.0& 80.91& 74.38& 3.32& 3.34& 53.58& \textbf{206.81}& 6.77& 90.76 \\
			& $p = 1.0$& 59.28& 31.57& 81.11& 74.68& 3.45& 3.49& 54.36& \textbf{206.81}& 7.59& 91.16 \\
			$\textbf{avg } d_v(S)$& $p = 1.05$& 59.25& 31.58& \textbf{81.12}& \textbf{74.69}& \textbf{3.67}& \textbf{4.15}& 54.36& \textbf{206.81}& 8.7& \textbf{91.18} \\
			& $p = 1.5$& 60.74& \textbf{31.61}& 80.8& 73.96& \textbf{3.67}& 4.0& \textbf{54.47}& 166.93& \textbf{9.56}& 88.86 \\
			& $p = 2.0$& \textbf{60.92}& 30.32& 78.99& 70.35& 3.62& 3.54& 20.39& 44.04& 2.65& 19.88 \\
			\midrule
			Avg & \emph{maxcore}& 3136.0& 841.0& 6335.5& 5685.5& 11.3& 11.7& 2182.4& 43840.3& 47.4& 9227.8 \\
			Squared & $p = 0.5$& 3297.6& 841.0& 7372.9& 7002.2& 11.3& 11.7& 2930.9& 43840.3& 47.4& 11486.8 \\
			Degree & $p = 1.0$& 4154.3& 1265.8& 7614.1& 7301.6& 12.2& 12.7& 3031.9& 43840.3& 59.2& 12146.5 \\
			& $p = 1.05$& 4226.3& 1269.0& 7624.9& 7336.1& 13.7& \textbf{19.2}& 3031.9& 43840.3& 189.3& 12220.1 \\
			$\textbf{avg } d_v(S)^2$& $p = 1.5$& 4691.7& 1356.1& 7776.6& 7691.7& 13.7& 17.3& 3051.2& 157225.3& 372.9& 14359.8 \\
			& $p = 2.0$& \textbf{5106.6}& \textbf{1384.1}& \textbf{7882.1}& \textbf{7918.9}& \textbf{13.9}& 18.6& \textbf{9730.9}& \textbf{455975.9}& \textbf{552.3}& \textbf{33262.3} \\
						\midrule
			Max & \emph{maxcore}& 56& 29& 153& 216& 7& \textbf{12}& 47& 391& 13& 447 \\
			Degree & $p = 0.5$& 108& 29& 196& 302& 7& \textbf{12}& 82& 391& 13& 844 \\
			& $p = 1.0$& 241& 161& 214& 333& 4& \textbf{12}& 84& 391& 10& 954 \\
			$\textbf{max } d_v(S)$& $p = 1.05$& 281& 163& 215& 338& 4& 8& 84& 391& 161& 978 \\
			& $p = 1.5$& 333& 200& 233& 399& 4& 8& 86& 5351& 106& 1428 \\
			& $p = 2.0$& \textbf{392}& \textbf{257}& \textbf{260}& \textbf{513}& \textbf{9}& \textbf{12}& \textbf{2329}& \textbf{34943}& \textbf{549}& \textbf{28754} \\
						\midrule
			Runtime & \emph{maxcore}& 0.03& 0.03& 0.06& 0.04& 1.37& 0.96& 2.4& 11.27& 0.45& 2.51 \\
			& $p = 0.5$& 0.29& 0.19& 0.38& 0.46& 4.8& 3.41& 51.28& 608.5& 2.01& 58.87 \\
			& $p = 1.0$& 0.03& 0.03& 0.06& 0.04& 1.37& 0.96& 2.4& 11.27& 0.45& 2.51 \\
			& $p = 1.05$& 0.25& 0.18& 0.35& 0.4& 4.4& 3.0& 41.92& 490.91& 1.62& 39.21 \\
			& $p = 1.5$& 0.26& 0.17& 0.32& 0.37& 4.29& 3.14& 39.64& 324.67& 1.67& 30.26 \\
			& $p = 2.0$& 0.25& 0.17& 0.32& 0.35& 4.12& 3.07& 19.13& 295.73& 1.69& 26.17 \\
			\bottomrule
	\end{tabular}}
\end{table*} 
\begin{figure}[t]
	\centering
	\subfloat[Polbooks, $n = 105$\label{fig:polbooks}] 
	{\includegraphics[width=.45\linewidth]{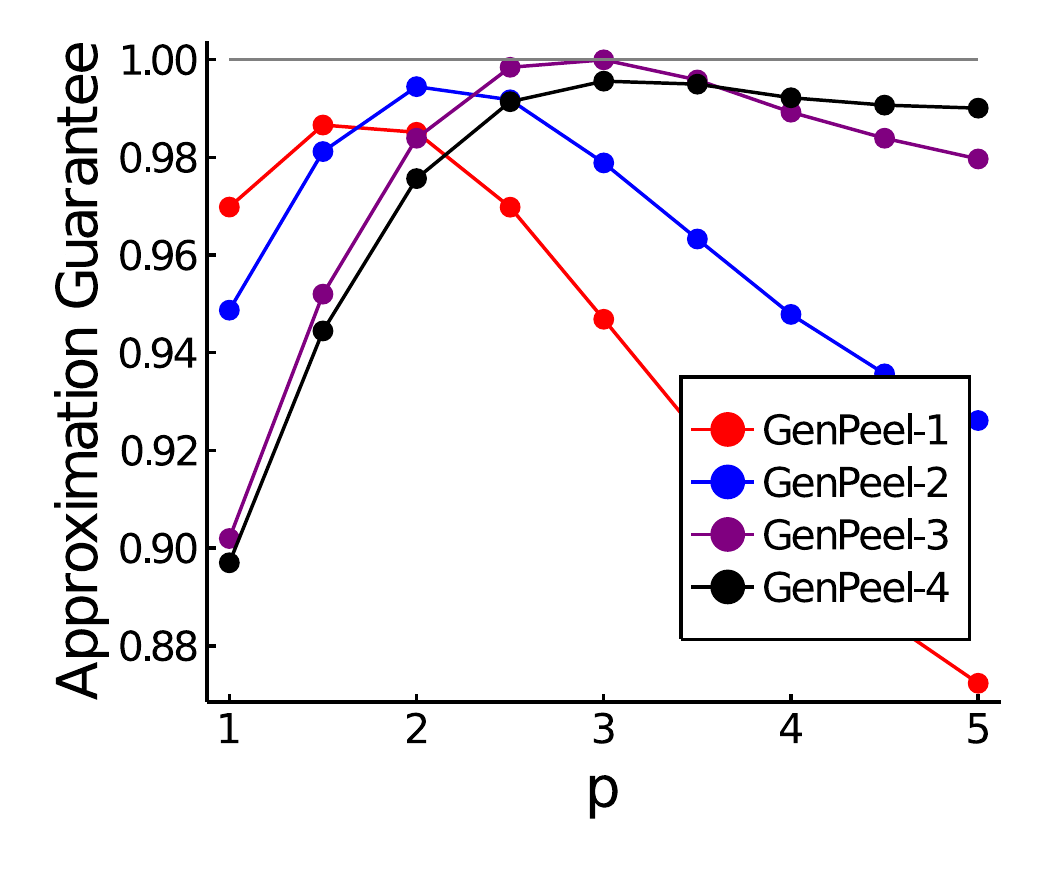}}\hfill
	\subfloat[Adjnoun, $n = 112$\label{fig:adjnoun}] 
	{\includegraphics[width=.45\linewidth]{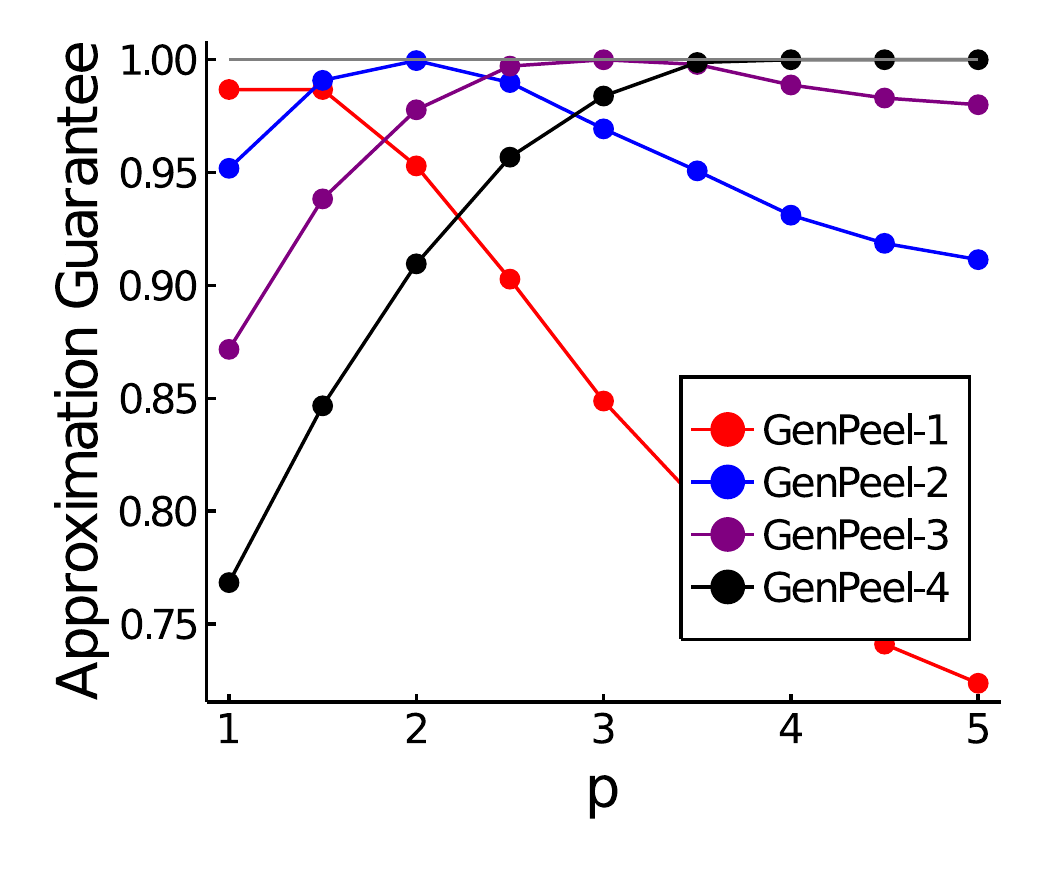}}\hfill
	\vspace{-.5\baselineskip}
	\subfloat[Jazz, $n = 198$ \label{fig:jazz}]
	{\includegraphics[width=.45\linewidth]{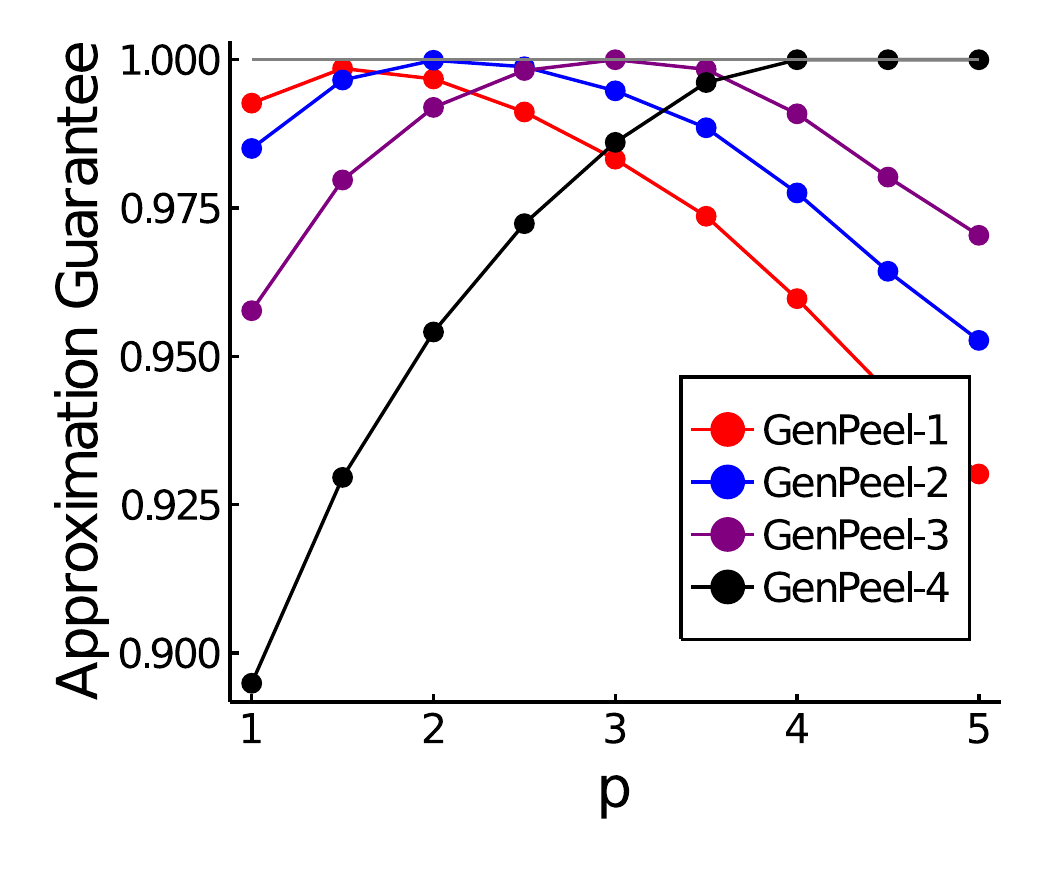}}\hfill
	\subfloat[Email, $n = 1005$\label{fig:email}] 
	{\includegraphics[width=.45\linewidth]{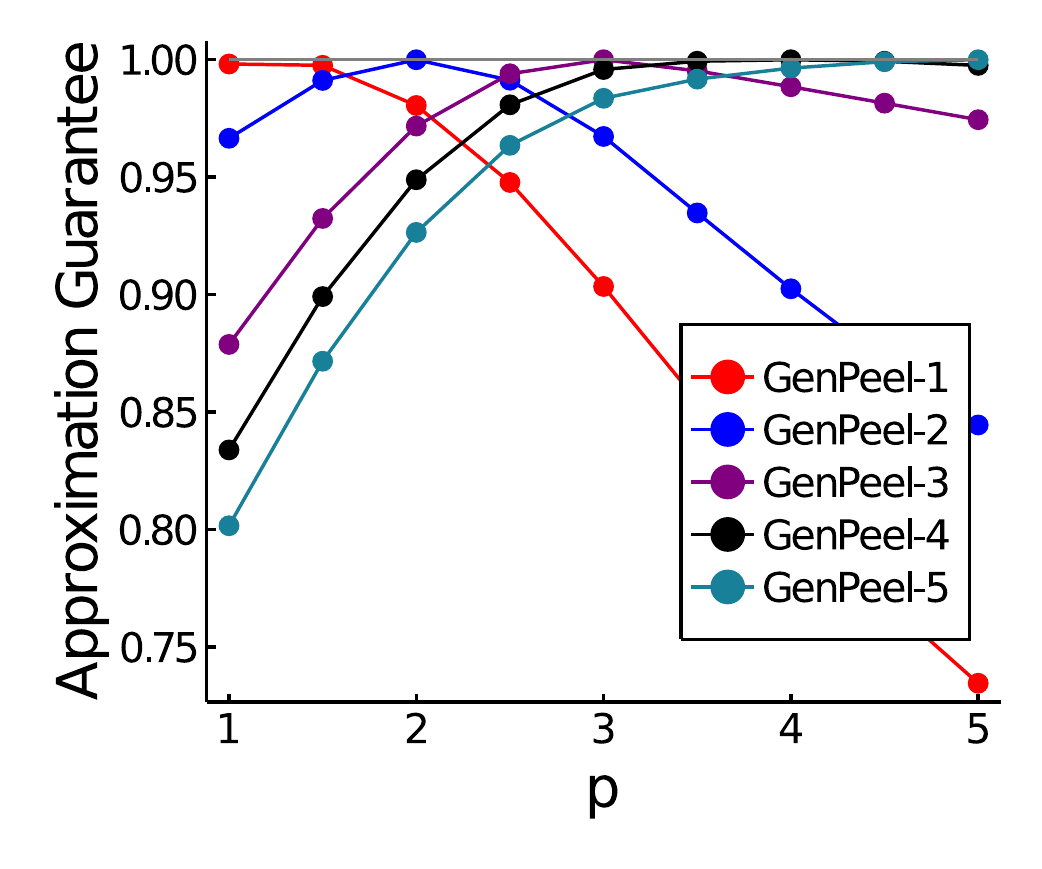}}\hfill
	\caption{Quality of the fast \textsc{GenPeel} heuristic compared to exact solution obtained with submodular minimization.
	Different runs of \textsc{GenPeel} for different values of $p$ do a good job of approximating the objective. 
	For plots (a)--(c), the $p = 4$ and $p = 5$ greedy solution are the same.
	}
	\label{fig:approx}
\end{figure}
In practice,~\textsc{GenPeel} can find a dense subgraph that approximates the optimal solution much better than its worst case guarantee. Furthermore, both the optimal $p$-mean densest subgraphs for different $p$, and the sets found by~\textsc{GenPeel} using different $p$, are meaningfully distinct from each other and highlight different notions of density in real-world graphs. To demonstrate this, we find a set of optimal solutions to our objective for $p$ values in $P_\mathit{objs} = \{1.0, 1.5, 2.0, \ldots, 5.0\}$. 
We solve our objective exactly on graphs with up to 1000 nodes to within a small tolerance with a MATLAB implementation that uses existing submodular minimization software as a black box~\cite{krause2010sfo}. 
We then run \textsc{GenPeel} for each $p \in P_\mathit{alg} =\{1.0,2.0,3.0,4.0,5.0\}$. This produces 5 dense subgraphs, and we evaluate how each one approximates the optimal solution for all $p \in P_\mathit{objs}$. 

  As expected, \textsc{GenPeel}-$p_0$ provides the best approximation for all $p$ near $p_0$. Rounded curves in Figure~\ref{fig:approx} show that each run of \textsc{GenPeel} optimizes a different regime of our problem. For three datasets, the $p \in \{4,5\}$ solutions are identical, both at optimality and for the sets returned by \textsc{GenPeel}. Otherwise, we see a clear distinction between the output curves for each run of the algorithm, indicating that we are finding different types of dense subgraphs.
  
\subsection{Peeling Algorithms for Dense Subgraphs}
Out next set of experiments places the standard greedy peeling algorithm for densest subgraph ($p = 1$) and the peeling algorithm for finding maxcore ($p = -\infty$) within a broader context of parametric peeling algorithms for dense subgraph discovery. 
By comparing these  outputs against \textsc{GenPeel} for different $p$ values near 1, we can observe how each method emphasizes and favors different notions of density in the graph. We can also see how running \textsc{GenPeel} for values near but not equal to one provides an accuracy vs.\ runtime tradeoff when it comes to finding sets that satisfy the traditional $p = 1$ notion of density.

We run \textsc{GenPeel} for $p \in \{0.5, 1.0, 1.05, 1.5, 2.0\}$. Although our method provides no formal guarantees when $p < 1$, it nevertheless greedily optimizes the $p$-mean density objective and produces meaningfully different subgraphs. When $p =1$, our implementation defaults to running the faster \textsc{SimplePeel} algorithm, and also outputs the maxcore solution, as this is obtained by simple peeling with a different stopping point. We focus on the regime $p \in [0.5,2]$ in order to better understand how different desirable measures of density vary as we explore above and below the standard densest subgraph regime ($p = 1$). Additionally, restricting to $p \leq 2$ means we are interpolating between objectives that have previously been studied~\cite{matula1983smallest,charikar2000greedy,farago2008general}, and avoids placing too high of an emphasis on just finding a small number of very high degree nodes.

Table~\ref{tab:snap} displays runtimes and subgraph statistics for each peeling result on a range of familiar benchmark graphs from a variety of different domains. This includes two citation networks (ca-Astro, condmat2005), two road networks (road-CA, road-TX), two web graphs (web-Google, web-BerkStan), an email network (Enron), two social networks (BrightKite, YouTube), and a retail graph (Amazon).
We report the edge density (number of edges divided by number of pairs of nodes), the size of the set returned, and the average degree (i.e., the $p = 1$ objective). There are clear trends in the output of each method: as $p$ decreases, the subgraphs tend to be smaller and have a higher edge density. 
As $p$ increases from 0.5 to 2, the average squared degree and the maximum degree increase significantly. What is perhaps most significant is that running \textsc{GenPeel}-$p$ with $p > 1$ tends to produce better sets than \textsc{SimplePeel} in terms of the standard densest subgraph objective. We conjecture that this is because \textsc{GenPeel} makes more strategic node removal decisions based not only on a node's degree, but also on the degree of its neighbor. This comes at the expense of a slower runtime, but we still find that our method is fast and scales up to very large graphs. 

\subsection{Dense Subgraphs in Social Networks}
\begin{figure}[t]
	\centering
	\includegraphics[width=.6\linewidth]{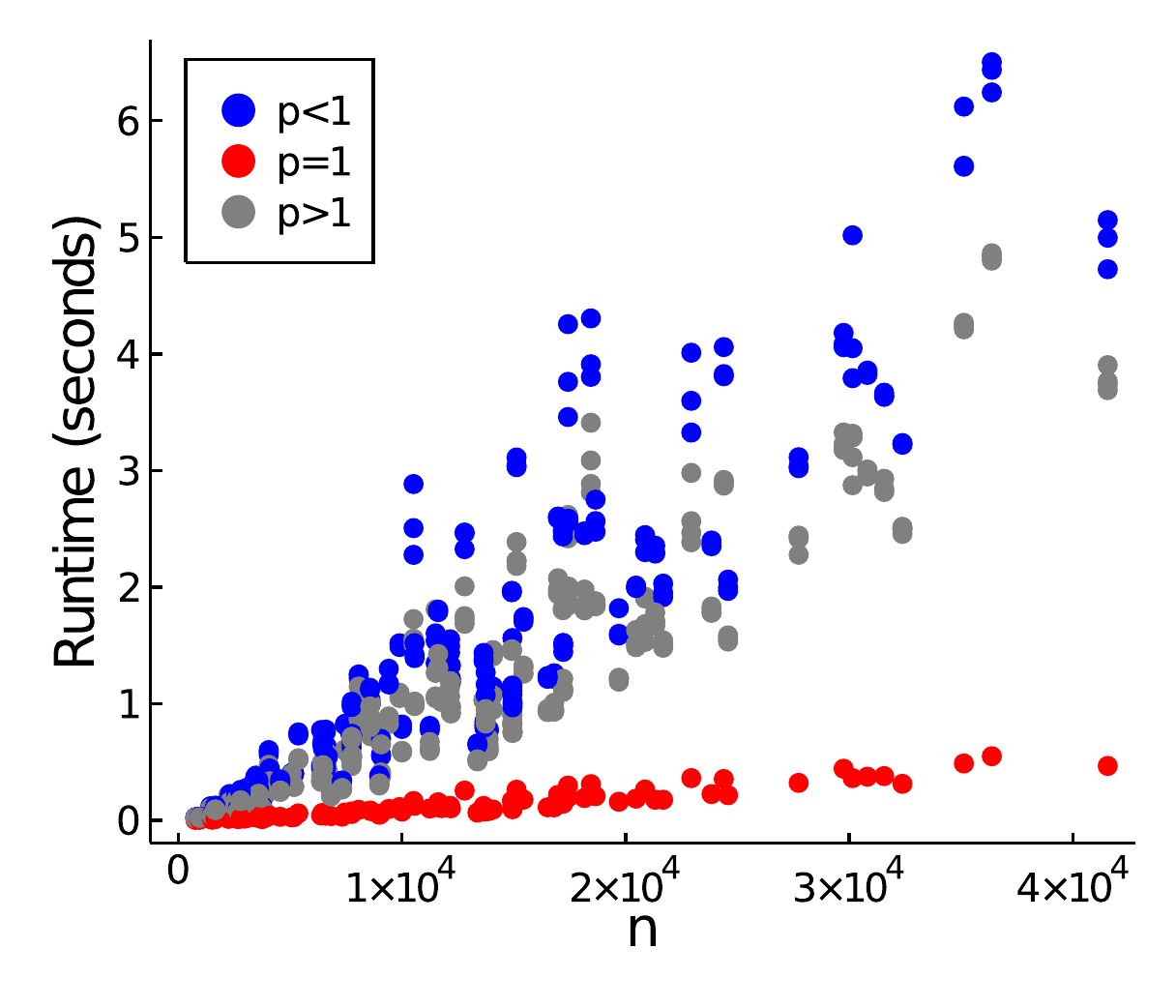}
	\caption{Runtime for Facebook100. \textsc{GenPeel} scales linearly on these datasets,
		and takes only a few seconds to run.}
	\label{fig:fbrun}
\end{figure}
Finally, we use \textsc{GenPeel} to detect different types of dense subgraphs in social networks. We run our method for $p \in \{0.25, 0.5, \ldots, 2.0\}$, 
on all graphs from the Facebook100 dataset~\cite{traud2012facebook}. Each graph is a snapshot of the Facebook network at a US university. The graphs have 700--42k nodes. 
For each of the 100 graphs, we plot curves indicating how the average degree, size, maximum degree, and edge density changes as $p$ increases (Figure~\ref{fig:fb}). In this parameter regime, the average degree of subgraphs returned hardly varies (Figure~\ref{fig:mean}). Increasing $p$ tends to produce subgraphs that are larger and have a higher maximum degree (Figures~\ref{fig:size} and~\ref{fig:max}), while decreasing $p$ produces smaller sets with higher edge density (Figure~\ref{fig:density}). We again find that running \textsc{GenPeel} with $p > 1$ can often find sets with higher average degree. This happens on six Facebook graphs when $p = 1.25$ (Figure~\ref{fig:mean}). We separately ran \textsc{GenPeel} with $p = 1.05$ on all graphs, and found that it returned sets with higher average degree than \textsc{SimplePeel} on 42 out of the 100 graphs.

Figure~\ref{fig:fbrun} is a scatter plot of points ($n$,$s$) where $n$ is the number of nodes in a Facebook graph and $s$ the the time in seconds for our algorithm. When $p = 1$, the method is much faster as it only considers node degree when greedily removing nodes. For the $p \neq 1$ case, runtime are still very fast and still scale roughly linearly, though with a steeper slope. For $p < 1$ and $p > 1$, the same exact procedure is applied. The $p < 1$ case nevertheless tends to take slightly longer, perhaps simply because greedy node removal in this regime can lead to more drastic changes in the function $\Delta_j$ that the algorithm must update at each iteration.
Still, on even the largest graphs, \textsc{GenPeel} takes just a few seconds.

\begin{figure}[t]
	\centering
	\subfloat[Average degree \label{fig:mean}] 
	{\includegraphics[width=.45\linewidth]{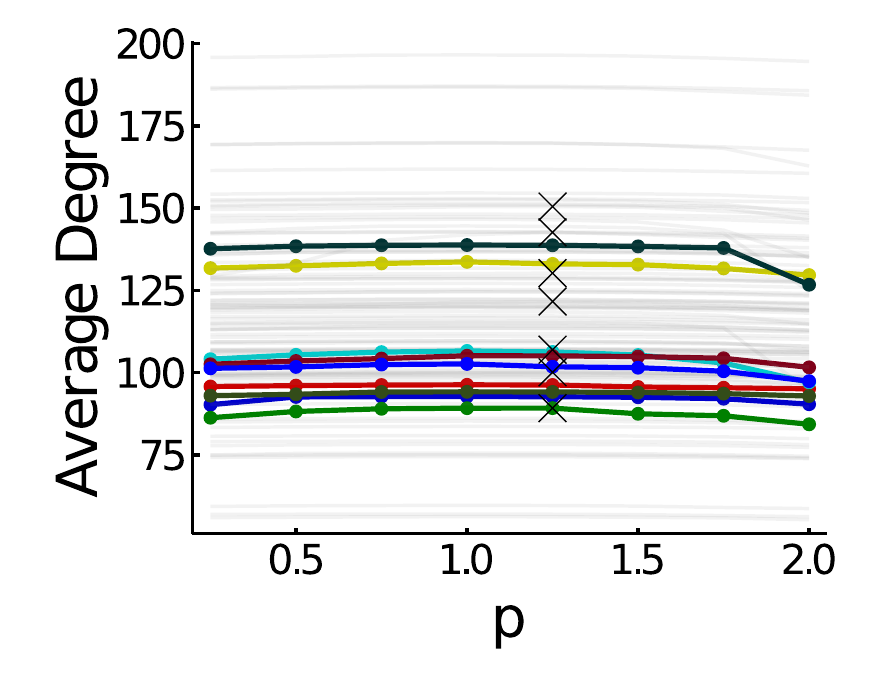}}\hfill
	\subfloat[Size\label{fig:size}] 
	{\includegraphics[width=.45\linewidth]{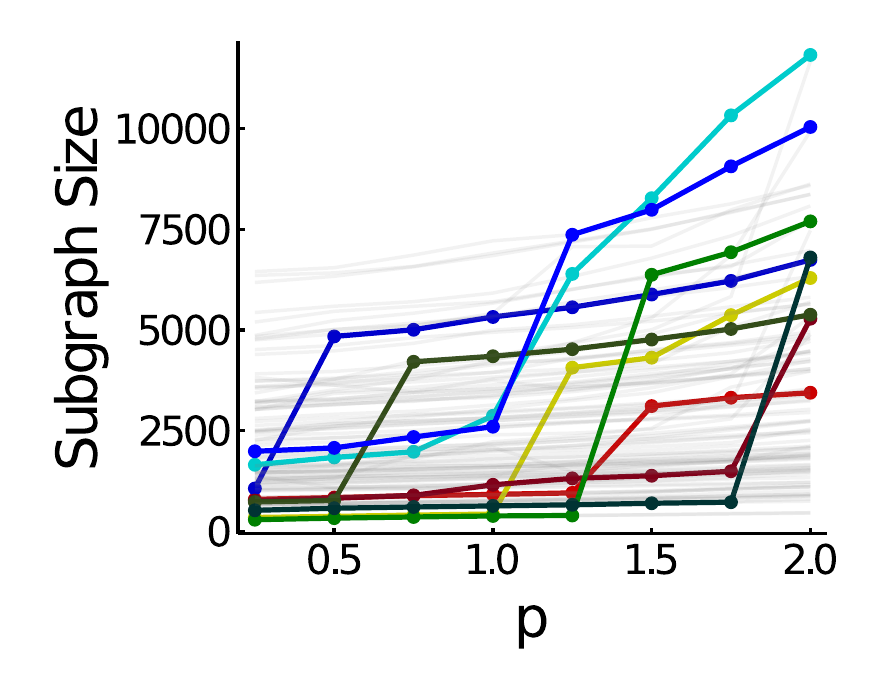}}\hfill
	\vspace{-.5\baselineskip}
	\subfloat[Max Degree \label{fig:max}]
	{\includegraphics[width=.45\linewidth]{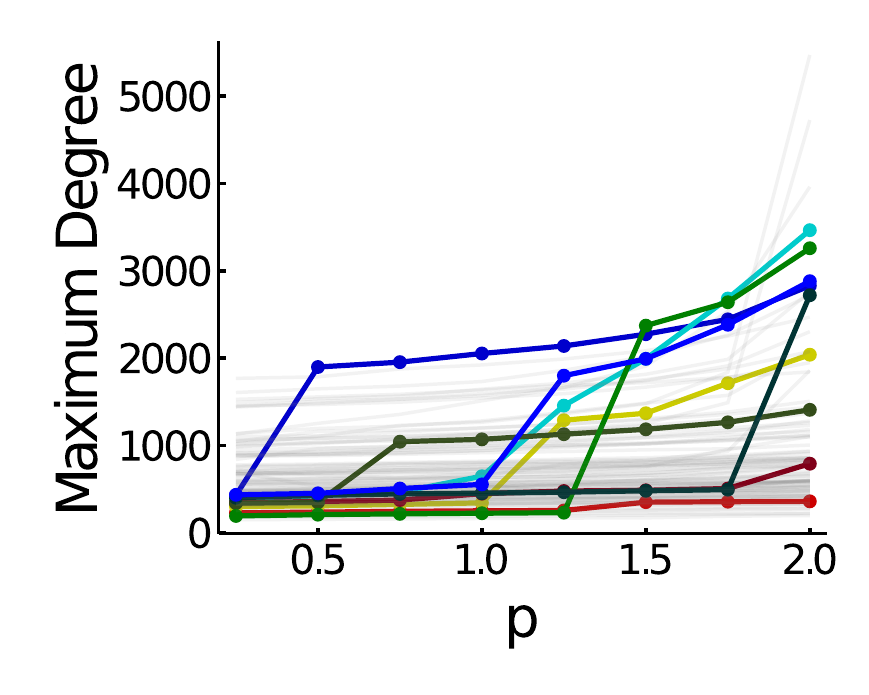}}\hfill
	\subfloat[Edge Density \label{fig:density}] 
	{\includegraphics[width=.45\linewidth]{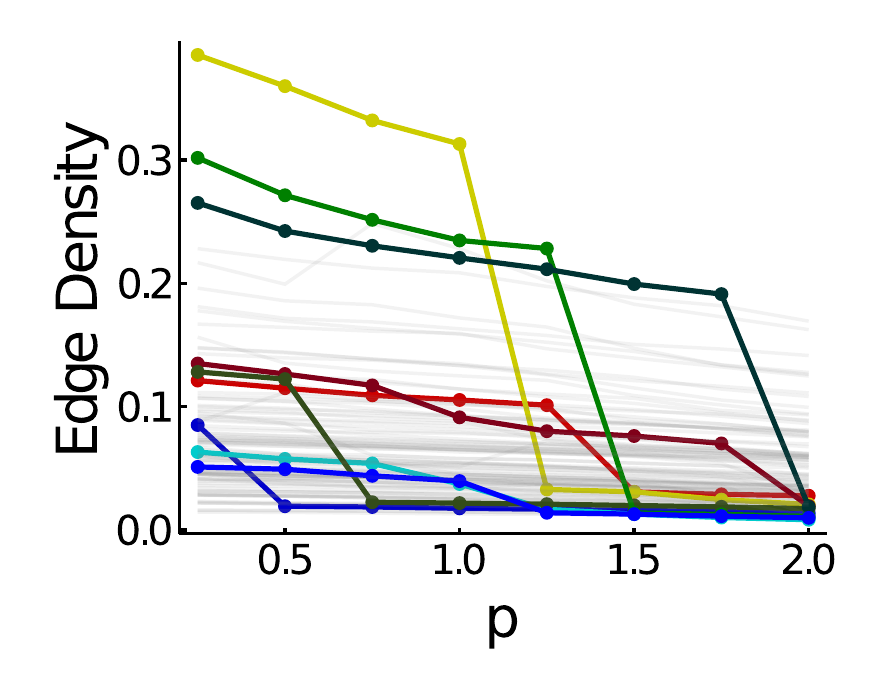}}\hfill
	\caption{We apply \textsc{GenPeel} on 100 subsets of Facebook for $p \in [0.25,2]$. Each line corresponds to one of the 100 FB networks, the ten networks that change the most in size are highlighted in color to illustrate key trends. There is significant overlap in sets; typically the subgraph obtained for a value $p$ contains most of the subgraph for $p_0 < p$. In this way our method is effectively uncovering a nested hierarchy of dense subgraphs that are related but whose properties change depending on the desired notion of density. (a) In general average degree changes very little, but the subgraphs vary in size (b), maximum degree (c), and edge density  $\nicefrac{|E_S|}{{|S| \choose 2}}$ (d), as we vary $p$. In several cases, we obtain a slightly higher average degree with \text{GenPeel} when $p = 1.25$ ($X$ marks in top left plot) than when $p = 1$, the case where our method reduces to the standard peeling algorithm.}
	\label{fig:fb}
\end{figure}

\section{Related Work}
\label{sec:related}
In order to situate our research within a broader context, we briefly review other related work on dense subgraph discovery. The standard densest subgraph problem is known to be polynomial time solvable via reduction to a maximum $s$-$t$ flow problem~\cite{gallo1989fast,goldberg1984finding}. When the goal is to find the densest subgraph on $k$ nodes, the problem becomes NP-hard~\cite{feige2001dense}. A number of methods for finding dense subgraphs focus specifically on greedy peeling algorithms. The standard peeling algorithm for dense subgraph discovery was first proposed by Ashahiro et al.~\cite{asahiro2000greedily} as a means to approximate the latter task, but was later shown by Charikar~\cite{charikar2000greedy} to provide a $1/2$-approximation to the unrestricted objective. This approximation was later adapted by Khuller and Saha~\cite{khuller2009finding} to provide an approximation in the case of directed graphs, a variant that can also be solved in polynomial time. A number of generalizations of the objective have recently been considered, including the $k$-clique densest subgraph~\cite{tsourakakis2015k}, the $\textbf{F}$-density objective~\cite{farago2008general}, the $f$-densest subgraph~\cite{kawase2018}, and generalizations involving signed graphs~\cite{tsourakakis2019novel}, bipartite graphs~\cite{sariyuce2018}, and fairness constraints~\cite{aris2020fairdense}.
Outside of direct variations on the densest subgraph problem that we generalized,
other formalisms for similar dense subgraphs include 
nucleus decompositions~\cite{sariyuce2015finding},
quasi-cliques~\cite{tsourakakis2013denser},
trusses~\cite{cohen2008trusses},
braces~\cite{ugander2012structural},
DN-graphs~\cite{wang2010triangulation},
plexes~\cite{seidman1978graph}, and
clubs and clans~\cite{mokken1979cliques}.
For additional related work, we refer to the survey by Lee et al.~\cite{lee2010survey} and a tutorial from Gionis and Tsourakakis~\cite{gionis2015dense}.

%
%

\section{Conclusion and Discussion}
Our $p$-mean densest subgraph objective unifies the standard densest subgraph and maxcore problems, and provides a general framework for capturing different notions of density in the same graph. We have presented several new algorithmic guarantees, including an optimal polynomial time solution based on submodular minimization, and a fast approximation algorithm based on a more sophisticated variant of the standard peeling method. The most compelling direction for future work is to develop computational complexity results and algorithms for $p \in (-\infty, 1)$. Experimental results using the greedy approximation for $p \in (0,1)$ indicate that this regime favors smaller clusters with a higher edge density, which is a particularly attractive property in dense subgraph discovery. A question that remains open even for $p \geq 1$ is whether a single peeling algorithm or alternative ordering method on the nodes can be used to define a nested set of dense subgraphs that can well approximate our objective for a wide range of $p$ values, simply by adding or subtracting additional nodes from the ordering as $p$ changes. In many of our experiments, the greedy method return subgraphs that were often nested or at least nearly nested. This therefore seems like a promising avenue for future theoretical results, or at least new practical techniques that provide a range of dense subgraphs without re-running an algorithm for each value of $p$. 

\begin{acks}
	This research was supported by NSF Award DMS-1830274, ARO Award W911NF19-1-0057, ARO MURI, JPMorgan Chase \& Co, a Vannevar Bush Faculty Fellowship, and a Simons Investigator grant.
\end{acks}

\bibliographystyle{ACM-Reference-Format}
\bibliography{pden-short}
\appendix
\section{Tightness of $p$-GenPeel}
For a graph $G = (V,E)$ and a fixed $p \geq 1$, define
\begin{equation}
\label{deltaG}
\delta(G) = \min_{v \in V} \,\, \left( (d_v)^p + \sum_{u \in \mathcal{N}(v)} d_u^p - (d_u-1)^p\right).
\end{equation}
This is the minimum change to the average $p$th power degree of $G$, and mirrors our definition of $\Delta_j(S)$ in~\eqref{Delta}. In order to find examples where the approximation guarantee of \textsc{GenPeel}-$p$ is tight, we would like to find graphs $G_1$ and $G_2$ such that (i) the $p$-mean densest subgraph of $G_i$ is all of $G_i$ for $i \in \{1,2\}$, (ii) $C = f(G_1)/f(G_2) \rightarrow (p+1)$, and (iii) $\delta(G_1) < \delta(G_2)$. 

If we can satisfy these three requirements, we can build a graph for which \textsc{GenPeel}-$p$ returns an approximation that asymptotically approaches $(p+1)$. To construct such an example, combine a single copy of $G_1$ with a large number of disjoint copies of $G_2$ to form a new graph $G$. The above properties guarantee that $G_1$ is the maximum $p$-densest subgraph of $G$, but the peeling algorithm will nevertheless remove all of its nodes before removing any node from any copy of $G_2$. As long as there are enough copies of $G_2$ (which we can add to $G$ without limit), the maximum $p$-density returned by \textsc{GenPeel}-$p$ will be roughly $f(G_2)$, for an overall approximation of $f(G_2)/f(G_1) = 1/(p+1)$.

\paragraph{Specific Graph Construction}
Let $r$ be an integer we will choose later and $G_2$ be a clique on $(r+1)$ nodes. Note that
\begin{align*}
\delta(G_2) &= r^p + r(r^p - (r-1)^p)\\
f(G_2) &= r^p.
\end{align*}
Next, let $G_1 = G_1(n,k)$ be a graph parameterized by integers $n$ and $k< n/2$, where $n$ is the number of nodes in $G_1$. For each node $i$, introduce an edge from $i$ to $(i+1)$, $(i+2), \hdots b_i$ where $b_i = \max\{i+k,n\}$. This means that most nodes will have degree $2k$, while nodes close to 1 and close to $n$ will have slightly smaller degree. Overall, as long as $n$ is large compared to $k$, we will have $f(G_1) \approx (2k)^p$. More precisely, we can calculate that
\begin{align*}
f(G_1) &= \frac{\sum_{i = 1}^k (k+i-1)^p + \sum_{i = k+1} (2k)^p + \sum_{i = n-k+1}^n (k+i-1)^p}{n} \\
&> \left(1 - \frac{2k}{n}\right) (2k)^p.
\end{align*}
Importantly, $\delta(G_1)$ is nearly the same as $f(G_1)$, and in fact $f(G_1)$ converges to $\delta(G_1)$ if $k$ is fixed and $n \rightarrow \infty$. In detail, note that the minimum degree node is node $1$, and so
\begin{align*}
\delta(G_1) &= k^p + \sum_{i = 2}^{k+1} (k+i-1)^p - (k+i-2)^p = (2k)^p. 
\end{align*}

Next, we need to choose a value of $r$ so that the greedy algorithm will opt to remove nodes from $G_1$ before removing nodes from one of the copies of $G_2$. This will happen as long as
\begin{equation*}
\delta(G_2) = r^p + r(r^p - (r-1)^p) > (2k)^p = \delta(G_1).
\end{equation*}
Choosing $r = \left \lceil \frac{2k}{(p+1)^{1/p}} + 1 \right \rceil$,  we have
\begin{align*}
r &\geq \frac{2k}{(p+1)^{1/p}} + 1 \implies (r-1)^p \geq \frac{(2k)^p}{p+1}\\
&\implies (p+1)(r-1)^p \geq (2k)^p = \delta(G_1).
\end{align*}
By Observation~\ref{obs:power}, we know that $p(r-1)^{p-1} \leq r^p - (r-1)^{p}$, and so
\begin{align*}
(p+1)(r-1)^p &= p(r-1)^p + (r-1)^p < pr(r-1)^{p-1} + r^p \\
&<  r (r^p - (r-1)^p) + r^p = \delta(G_2).
\end{align*}
To compute the asymptotic approximation guarantee for \textsc{GenPeel}-$p$, note that $r < 2k/(p+1)^{1/p} + 2$ and so
\begin{align*}
\frac{f(G_1)}{f(G_2)} &> \frac{(1-\frac{2k}{n})(2k)^p}{r^p} > \left(1-\frac{2k}{n}\right)\frac{(2k)^p}{ \left[\frac{2k + 2(p+1)^{1/p}}{(p+1)^{1/p}} \right]^p} \\
&= \left(1-\frac{2k}{n}\right)\frac{(p+1)}{ \left[1+\frac{(p+1)^{1/p}}{k} \right]^p} 
\end{align*}
Therefore, for any fixed finite $p \geq 1$, if $k$ satisfies $k = o(n)$ and $\frac{1}{k} \rightarrow 0$, then this overall quantity converges to $(p+1)$. This means that asymptotically, the average $p$th power degree of $G_1$ is $(p+1)$ times better than the average $p$th power degree of $G_2$, so this is the best approximation guarantee we can obtain after $G_1$ has been deleted. The set with the best density considered by \textsc{GenPeel}-$p$ will be the entire graph $G$, since the density will be slightly better before we delete $G_1$. However, we still have the same asymptotic approximation guarantee, since we can include $n$ copies of $G_2$ when forming $G$, making  $f_p(G)$ asymptotically close to $f_p(G_2)$. For small values of $p$ (e.g., integers up to 10), it is not hard to numerically find examples of graphs with approximation guarantees between $p$ and $p+1$, using this type of graph construction.

\end{document}